\documentclass[11pt]{article}

\usepackage[T1]{fontenc}

\usepackage{amsfonts}  
\usepackage{url}
\begin{document}
\title{Arithmetic Circuits:\\ The Chasm at Depth Four Gets Wider}

\author{Pascal Koiran\\
LIP\thanks{UMR 5668 ENS Lyon, CNRS, UCBL, INRIA.}, \'Ecole Normale Sup\'erieure de Lyon, Universit\'e de Lyon\thanks{This work was done during a visit to the Fields Institute and to the University of Toronto's Department of Computer Science.}\\
{\tt Pascal.Koiran@ens-lyon.fr} 
}

\date{\today}
\maketitle

\newtheorem{conjecture}{Conjecture}
\newtheorem{theorem}{Theorem}
\newtheorem{lemma}{Lemma}
\newtheorem{proposition}{Proposition}
\newtheorem{corollary}{Corollary}
\newtheorem{definition}{Definition}
\newtheorem{problem}{Problem}
\newtheorem{remark}{Remark}
\newtheorem{example}{Example}
\newtheorem{hypothesis}{Hypothesis}

\makeatletter
\def\@yproof[#1]{\@proof{ #1}}
\def\@proof#1{\begin{trivlist}\item[]{\em Proof#1.}}
\newenvironment{proof}{\@ifnextchar[{\@yproof}{\@proof{} 
}}{~$\Box$\end{trivlist}}
\makeatother

\newcommand\cc{\ensuremath{\mathbb{C}}}

\newcommand\vpspace{\ensuremath{\mathsf{VPSPACE}}}
\newcommand\pspace{\ensuremath{\mathsf{PSPACE}}}
\newcommand\vpzero{\ensuremath{\mathsf{VP}^0}}
\newcommand\rr{\ensuremath{\mathbb{R}}}
\newcommand\vp{\ensuremath{\mathsf{VP}}}
\newcommand\vpnb{\ensuremath{\mathsf{VP}_{\mathsf{nb}}}}
\newcommand\vpnbzero{\ensuremath{\mathsf{VP}^0_{\mathsf{nb}}}}
\newcommand\vnp{\ensuremath{\mathsf{VNP}}}
\newcommand\vnpzero{\ensuremath{\mathsf{VNP}^0}}
\newcommand\vnpnb{\ensuremath{\mathsf{VNP}_{\mathsf{nb}}}}
\newcommand\vnpnbzero{\ensuremath{\mathsf{VNP}^0_{\mathsf{nb}}}}
\newcommand\vpip{\ensuremath{\mathsf{V\Pi P}}}
\newcommand\vpipzero{\ensuremath{\mathsf{V\Pi P}^0}}
\newcommand\poly{\ensuremath{\mathsf{poly}}}
\newcommand\zz{\ensuremath{\mathbb{Z}}}
\newcommand\nn{\ensuremath{\mathbb{N}}}
\newcommand\fp{\ensuremath{\mathsf{FP}}}
\newcommand\p{\ensuremath{\mathsf{P}}}
\newcommand\pnu{\ensuremath{\mathbb{P}}}
\newcommand\np{\ensuremath{\mathsf{NP}}}
\newcommand\npnu{\ensuremath{\mathbb{NP}}}
\newcommand\nc{\ensuremath{\mathsf{NC}}}
\newcommand\nl{\ensuremath{\mathsf{NL}}}
\newcommand\logcfl{\ensuremath{\mathsf{LOGCFL}}}
\newcommand\per{\ensuremath{\mathrm{PER}}}
\newcommand\sharpp{\ensuremath{\mathsf{\sharp P}}}
\newcommand\pp{\ensuremath{\mathsf{PP}}}
\newcommand\gapp{\ensuremath{\mathsf{GapP}}}
\newcommand\gapppoly{\ensuremath{\mathsf{GapP/poly}}}
\newcommand\chpoly{\ensuremath{\mathsf{CH/poly}}}
\newcommand\ppoly{\ensuremath{\mathsf{P/poly}}}
\newcommand\ch{\ensuremath{\mathsf{CH}}}

\newcommand\bit{\ensuremath{\mathrm{Bit}}}

\newcommand\hit{\ensuremath{\cal H}}

\begin{abstract}
In their paper on the ``chasm at depth four'', 
Agrawal and Vinay have shown that
polynomials in $m$ variables of degree $O(m)$ 
 which admit arithmetic circuits of size $2^{o(m)}$ also admit
 arithmetic circuits of depth four and size $2^{o(m)}$. 
This theorem shows that for problems such
as arithmetic circuit lower bounds or black-box derandomization 
of identity testing, the case of depth four circuits is in a certain sense
the general case.

In this paper we show that smaller depth four circuits can be obtained
if we start from polynomial size arithmetic circuits.
For instance, we show that  if the permanent of 
$n \times n$ matrices has circuits of size polynomial in $n$, then it also
has depth 4 circuits of size $n^{O(\sqrt{n}\log n)}$.
If the original circuit uses only integer constants of polynomial size, 
then the same is true of the resulting  depth four circuit.
These results have potential applications to lower bounds and deterministic 
identity testing, in particular for sums of products of sparse univariate 
 polynomials.
We also use our techniques to reprove two results on:
\begin{itemize}
\item[-] The existence of nontrivial boolean circuits of constant depth for
languages in $\logcfl$.
\item[-]  Reduction to polylogarithmic depth
for arithmetic circuits of polynomial size and polynomially bounded degree.
\end{itemize}
\end{abstract}

\newpage
\section{Introduction}

Agrawal and Vinay have shown that
polynomials of degree $d=O(m)$ in $m$ variables
 which admit nontrivial arithmetic circuits also admit nontrivial
 arithmetic circuits 
of depth four~\cite{AgraVinay08}. Here, ``nontrivial'' means of size 
$2^{o(d+d\log{m \over d})}$. The resulting depth 4 circuits are $\sum \prod \sum \prod$ arithmetic formulas: the output gate (at depth 4) and the gates at depth 2 are addition gates,  and the other gates are multiplication gates.
This theorem shows that for problems such
as arithmetic circuit lower bounds or black-box derandomization 
of identity testing, the case of depth four circuits is in a certain sense
the general case.

But what if we start from arithmetic circuits of size 
smaller than $2^{o(m)}$ 
(for instance, of size polynomial in $m$)?
It is reasonable to expect that the size of the corresponding 
depth four circuits will be reduced accordingly, but such a result
cannot be found in~\cite{AgraVinay08}.
One of the main results of this paper is a depth reduction theorem for $\vp$ families (i.e., families $(f_n)$ of polynomials of degree
and arithmetic circuit complexity polynomially bounded in $n$).
We show in Theorem~\ref{vp2depth4} that any $\vp$ family $(f_n)$ has
depth 4 arithmetic formulas of size $n^{O(\sqrt{d_n}\log d_n)}$, where $d_n$ is the degree of $f_n$. For instance, this result shows that if the permanent of 
$n \times n$ matrices has circuits of size polynomial in $n$, then it also
has depth 4 formulas of size $n^{O(\sqrt{n}\log n)}$.
This is potentially useful for a lower bound proof:
to show that the permanent does not have polynomial size circuits, we ``only''
have to show that it does not have depth 4 formulas of size  
$n^{O(\sqrt{n}\log n)}$. 
This is still certainly far away from the known lower bounds for constant depth arithmetic circuits: currently we have superpolynomial lower bound for the
permanent for circuits of depth 3 only, and only in finite fields~\cite{GriKar98,GriRaz00}. 
In the restricted setting of multilinear arithmetic circuits, superpolynomial lower bounds
can be obtained for circuits of arbitrary constant depth~\cite{RY09}.
We do not address the issue of multilinearity in this paper. 
Note however  that
the results in~\cite{RY08,RY09} suggest that the bound in Theorem~\ref{vp2depth4}
could be fairly close to optimal at least for multilinear circuits.
Indeed, a polynomial $f$ of degree $3n-1$ in $O(n^3)$ variables
with multilinear arithmetic circuits of polynomial size
is constructed in Section~4 of~\cite{RY08}.
By Theorem~4.3 of~\cite{RY08} and Theorem~5.1 of~\cite{RY09},
 all multilinear depth~4 circuits for $f$ are of size 
at least $n^{\Omega(\sqrt{n/\log(n)})}$. 
This shows that the exponent $\sqrt{d_n}$ in Theorem~\ref{vp2depth4} 
cannot be removed if we insist on a reduction to depth 4 that would preserve
multilinearity. Note that for reduction to depth $\log^2(n)$, preservation
of multilinearity is indeed possible~\cite{RY08}.

We also perform an analysis of the size of the integer 
constants used by the depth 4 circuit simulating a given polynomial size
circuit (a similar analysis for the construction in~\cite{AgraVinay08} 
has not been carried out yet to the author's knowledge). 
Roughly speaking, we show that reduction to depth 4 does not require
the introduction of large constants. In particular, 
we give in Theorem~\ref{vpzero2depth4} an analogue of
Theorem~\ref{vp2depth4} for $\vp^0$ (this is a constant-free version of 
$\vp$). This result is used in~\cite{Koi10a}, where we show that black-box derandomization of identity testing for sums of products of sparse univariate polynomials with sparse coefficients would imply a lower bound for the permanent.
Finally, we give applications of our depth reduction techniques to boolean
circuit complexity and to the construction of arithmetic circuits of polylogarithmic depth.

\subsection{Main Ideas and Comparison with Previous Work}

The main depth reduction result in~\cite{AgraVinay08} is as follows.
\begin{theorem} \label{AgraVinay}
Let $P(x_1,\ldots,x_m)$ be a polynomial of degree $d=O(m)$ over a field $F$.
If there exists an arithmetic circuit of size $2^{o(d+d\log{m \over d})}$
for $P$ then there exists a depth 4 arithmetic circuit of size 
$2^{o(d+d\log{m \over d})}$. 
\end{theorem}
Theorem 2.4 in~\cite{AgraVinay08} also provides some bounds on the fan-in
of the gates in the resulting depth 4 circuits.

For multilinear polynomials, 
their result (Corollary~2.5 in~\cite{AgraVinay08}) 
reads as follows:
\begin{corollary} \label{reduction}
A multilinear polynomial in $m$ variables which has an arithmetic circuit of size $2^{o(m)}$ 
also has a depth 4 arithmetic circuit of size $2^{o(m)}$.
\end{corollary}
We give the (simple) proof, which is omitted from~\cite{AgraVinay08}.
For $d=m$ the result is clear since the exponent $d+d\log{m \over d}$ 
in Theorem~\ref{AgraVinay} is equal to $m$.
Consider now the case of a polynomial $P(X_1,\ldots,X_m)$ of degree $d<m$,
having a circuit of size $2^{o(m)}$. Let $Q=P+\prod_{i=1}^m X_i$. 
Since the number of variables of $Q$ is equal to its degree, we are back to 
the first case: $Q$ has a depth four circuit of size $2^{o(m)}$.
We can obtain a circuit of size $2^{o(m)}$ for $P$ by subtracting the product
$\prod_{i=1}^m X_i$ (this requires only $m$ additional arithmetic operations).
Note that this corollary and its proof hold more generally for any 
(possibly not multilinear) polynomial
of degree $d \leq m$.

By specializing the multilinear polynomial to the permanent, 
Agrawal and Vinay then state in Corollary~2.6 that if every depth 4 
arithmetic circuit for the permanent requires exponential size,
the same is true for arithmetic circuits of unbounded depth.
It is not made precise in~\cite{AgraVinay08} what ``exponential size'' 
exactly means. 
In this context (arithmetic complexity of the permanent) 
the most standard interpretation is probably that an exponential size
circuit for the $n \times n$ permanent is of size $2^{\Omega(n)}$
(note that the number of variables is $m=n^2$).
With this interpretation, it is not clear why Corollary 2.6 
of~\cite{AgraVinay08} would follow from Theorem~\ref{AgraVinay} or Corollary~\ref{reduction}.

Since the permanent of a $n \times n$ matrix has degree $d=n$ and 
$m=n^2$ variables, we can deduce the following from Theorem~\ref{AgraVinay}:
If there exists an arithmetic circuit of size $2^{o(n \log n)}$
for the $n \times n$ permanent  
then there exists also a depth 4 arithmetic circuit of size 
$2^{o(n \log n)}$. 
This statement is not very useful since we already know
(by Ryser's formula~\cite{ryser}) that the permanent has depth 3 arithmetic 
formulas of size $O(n2^n)$.
Note that applying Corollary~\ref{reduction} directly to the permanent
would give an even worse bound (namely, we would obtain depth 4 formulas 
of size $2^{o(n^2)}$).
As explained earlier, we can show that if the permanent has polynomial size 
circuits it must
also have depth 4 formulas of size $n^{O(\sqrt{n}\log n)}$.
This result does not follow from Theorem~\ref{AgraVinay}.
On the other hand, our results are weaker than Theorem~1  if we start from
a very large circuit. Indeed, as explained below, we can only show that 
a circuit of size $t$ and degree $d$ has an equivalent depth 4 circuit 
of size $t^{O(\sqrt{d} \log d)}$. This does not imply Theorem~\ref{AgraVinay}.

Before describing their general depth reduction algorithm, Agrawal and Vinay
begin with the special case of matrix powering. For this problem there is
a very simple and elegant reduction to depth four. 
Then they treat the  general case with an apparently different approach:
their construction builds on the depth reduction algorithm of Allender, Jiao, Mahajan and Vinay~\cite{AJMV98}, 
who gave a uniform version of the depth reduction result due to Valiant, Skyum, Berkowitz and Rackoff~\cite{VSBR83}. 
In this paper we show that the matrix powering idea is powerful enough to 
handle arbitrary polynomial-size arithmetic circuits. 
Arithmetic branching programs and weakly skew circuits
are the main tools that we use to reduce the evaluation 
of arbitrary arithmetic circuits to matrix powering.
These models are known to capture the complexity of a number
a problems from linear algebra such as e.g. matrix powering, iterated matrix multiplication or computation of the determinant~\cite{Toda92,MP08}.

\subsection{Organization of the paper}

In Section~\ref{arith} we present the two main computation models 
that we will use: arithmetic circuits and arithmetic branching programs.
We define some of the corresponding complexity classes, and give
some basic properties.
In Section~\ref{tobp}, building on a construction of Malod and Portier~\cite{MP08}  we give an efficient simulation of arithmetic
circuits by arithmetic branching programs. 
Compared to~\cite{MP08}, we take extra care to construct branching programs
of small depth because the square root of the depth appears in 
the exponent of the size estimate for the final depth 4 circuit.
In section~\ref{todepth4} we reduce branching programs to depth 4 circuits
using the matrix powering idea from~\cite{AgraVinay08}.
Then we state our main technical result in Theorem~\ref{circuit2depth4}.
We show in particular that an arithmetic circuit of size $t$ and formal degree $d$ has a depth 4 circuit of size $t^{O(\sqrt{d} \log d)}$.
We draw some consequences for depth reduction of $\vp$ families in Section~\ref{vp}, and for depth reduction of $\vpzero$ families in Section~\ref{vp0}.

In Section~\ref{logcfl} we give an application of these techniques to 
boolean circuit complexity. Namely, we show that languages in $\logcfl$
have constant-depth boolean circuits of size $2^{n^{\epsilon}}$
(and we briefly present the history of this result).

Finally, we show in Section~\ref{polylog} that the same tools can be used
to give a very simple (but suboptimal) proof of the fact that
for circuits of polynomially bounded size and degree, 
reduction to polylogarithmic depth can be achieved while preserving
polynomial size~\cite{VSBR83}.

\section{Arithmetic Circuits and Branching Programs} \label{arith}

We recall that an  arithmetic circuit contains addition and multiplication gates. In addition to these arithmetic gates there are input gates,
labelled by variables or constants from some field $K$.
An output gate is of fan-out zero. We often assume that there
is a single ouptut gate.
In this case an arithmetic circuit therefore represents a polynomial 
with coefficients in $K$.
Without loss of generality, we can and will assume  that every input gate has fan-out at most 1 (several input gates
can be labeled with the same variable or constant if necessary).

We often assume that the arithmetic gates have arity 2,
but in constant-depth circuits we naturally  allow addition and multiplication gates of unbounded fan-in (we often also some explicit upper bounds
on the fan-in, see for instance Theorem~\ref{circuit2depth4}).
In some of our intermediate constructions (e.g. Proposition~\ref{tows}) 
we also work with weighted addition
gates.
\begin{definition} \label{weighted}
A $n$-ary weighted addition gate computes a linear combination 
$a_1x_1+\cdots+a_nx_n$ of its inputs $x_1,\ldots,x_n$. Here $a_i$ is the weight associated to the $i$-th input of the gate. The total weight of the gate is 
$\sum_{i=1}^n |a_i|$.
\end{definition}
For instance, a subtraction gate is a binary weighted addition gate with 
weights $(1,-1)$.  We sometimes refer to binary unweighted addition gates as ``ordinary addition gates''.
The size of a circuit is its total  number of gates (including input gates).
\begin{definition}
Fix a field $K$.
A sequence $(f_n)$ of polynomials with coefficients in $K$
belongs to $\vp$ if there exists a polynomial $p(n)$ 
and a sequence $(C_n)$ of arithmetic circuits 
such that $\deg(f_n) \leq p(n)$, 
$C_n$ computes $f_n$ and is of size at most $p(n)$.
\end{definition}
The size constraint implies in particular that $f_n$ depends on
polynomially many variables.
The above definition is fairly robust. For instance we obtain the same class
with circuits using gates of fan-in 2 or of unbounded fan-in, 
weighted or unweighted addition gates.

An arithmetic formula is a circuit where all gates are of fan-out one, except
of course the output gate. In the constant depth setting, arithmetic formulas
and arithmetic circuits are polynomially related (\cite{RY09}, Claim~2.2).

The complexity of several problems from linear algebra
such as iterated matrix multiplication or computing the determinant 
is captured by a restricted class of arithmetic circuits called 
{\em weakly skew circuits}~\cite{Toda92,MP08}.
Let $C$ be an arithmetic circuit where all multiplication gates are binary.
A multiplication gate $\alpha$ in $C$ is said to be disjoint if at least
one of its two subcircuits is disjoint from the remainder of $C$, 
except of course for the edge from the subcircuit to~$\alpha$ 
(removing this edge would therefore disconnect $C$).
The circuit is weakly skew if its multiplication gates are all disjoint.
This definition is usually given only for circuits where all addition gates
are binary unweighted, but we will use our slightly more general definition
instead (see Propositions~\ref{tows} and~\ref{skew2abp}).

There is also a closely related notion of {\em skew circuits}~\cite{Toda92,Jansen08,KaKoi08}: 
a circuit with binary multiplication gates 
is skew if for every multiplication gate at least one of the two incoming edges
comes from an input of the circuit. 
Since we have assumed that that input gates have fan-out at most 1, every
skew circuit is also weakly skew.

A circuit where the only constants are from the set $\{0, -1, 1\}$ is said to 
be constant-free.
A constant-free circuit represents a polynomial in $\zz[X_1,\ldots,X_n]$, where
$X_1,\ldots,X_n$ are the variables labelling the input gates.

The constant-free model was systematically studied by Malod~\cite{Malod03}.
In particular, he defined a class $\vp^0$ of polynomial families
that are ``easy to compute'' by constant-free arithmetic circuits.
First we need to recall the notion of formal degree:
\begin{itemize}
\item[(i)] The formal degree of an input gate is equal to 1.
\item[(ii)]
The formal degree of an addition gate is the maximum of
the formal degrees of its incoming gates, and the  formal degree
of a multiplication gate is the sum of these formal degrees.
\end{itemize}
Finally, the formal degree of a circuit is equal to the formal degree
of its output gate.
This is obviously an upper bound on the degree of the polynomial
computed by the circuit.
Note that this definition can be applied to circuits with weighted addition gates of arbitrary fan-in. For instance, the polynomial $x-2y$ can be computed by a circuit containing one ordinary addition gate, one multiplication gate 
and three inputs labeled by $x$, $y$ and the constant $-2$. 
This circuit has formal degree two. 
The same polynomial can be computed by another 
circuit containing a binary weighted adition gate (of total weight $1+|-2|=3$)
with inputs $x$ and $y$. The second circuit has formal degree 1.
\begin{definition}
A sequence $(f_n)$ of polynomials 
belongs to $\vp^0$ if there exists a polynomial $p(n)$ 
and a sequence $(C_n)$ of
constant-free arithmetic circuits (with unweighted addition gates)
such that $C_n$ computes $f_n$ and is of size 
and formal degree at most $p(n)$.
\end{definition}
 The constraint on the formal degree 
forbids the computation of polynomials of high degree such as e.g. $X^{2^n}$;
it also forbids the computation of large constants such as $2^{2^n}$.
The class $\vp^0$ is therefore a strict subset of $\vp$ (over the field of rational numbers, or more generally any field of characteristic 0).
As for $\vp$ we obtain the same class with gates of fan-in 2 or of unbounded fan-in, but of course we cannot allow addition gates with arbitrary weights.
We can however allow subtraction gates:
\begin{proposition} \label{nosub}
Let $C$ be a
constant-free circuit of size $t$ and  formal
degree $d$, where the arithmetic gates are multiplication, 
unweighted addition or subtraction gates (all of fan-in 2).

There is an equivalent constant-free circuit $C'$ of formal
degree $d+1$ and size at most $6t+3$, 
where the arithmetic gates are binary multiplications 
or ordinary additions.
\end{proposition}
\begin{proof}
We need to get rid of subtraction gates. A first idea would be to write each
subtraction $x-y$ as $x+(-1)\times y$, but the cumulative effect of the multiplications $(-1)\times y$ could lead to an increase in the formal degree by more than 1. Instead we will represent each gate $\alpha$ in $C$ by a pair 
of gates $(\alpha_1,\alpha_2)$ in $C'$. The output of $\alpha$ will be equal to the differences of the outputs of $\alpha_1$ and $\alpha_2$.
An input $x$ in $C$ can be represented by the pair $(x,0)$.
To simulate the arithmetic operations in $C$
we use the following rules: 
$(\alpha_1-\alpha_2)+(\beta_1-\beta_2)=(\alpha_1+\beta_1)-(\alpha_2+\beta_2)$;
$(\alpha_1-\alpha_2)-(\beta_1-\beta_2)=(\alpha_1+\beta_2)-(\alpha_2+\beta_1)$;
$(\alpha_1-\alpha_2)\times(\beta_1-\beta_2)=
(\alpha_1\times \beta_1+\alpha_2\times\beta_2)-(\alpha_2\times \beta_1+\alpha_1\times\beta_2)$.
A straightforward induction shows that the gates in a pair 
$(\alpha_1,\alpha_2)$ will have same formal degree as the gate $\alpha$ that
they represent. Finally, to complete the construction of $C'$ we come back to our first idea: if $(\alpha_1,\alpha_2)$ is the pair representing the
output gate of $C$, we write the difference $\alpha_1 - \alpha_2$ as 
$\alpha_1 +(-1)\times \alpha_2$. This increases the formal degree by 1.
Each arithmetic operation in $C$ is simulated by at most 6 operations in $C'$, 
and we need 3 additional gates to perform the final subtraction.
\end{proof}
This modest increase in the formal degree  cannot be avoided: 
without subtraction gates there is no better way to compute the polynomial
$f(x)=-x$ than by the formula $f(x)=-1 \times x$, which is of formal 
degree 2.

Finally we define the notion of arithmetic branching program. 
This is an edge-weighted directed acyclic graph with two distinguished vertices $s$ and~$t$. The output of the branching program is by definition equal to the sum of the weights of all paths from $s$ to $t$, where the weight of a path is the product of the weights of its edges. In this paper we assume that the edge weights are constants from some field $K$ or variables. Like an arithmetic circuit, a branching program therefore represents a polynomial with
coefficients in~$K$.
The depth of a branching program is the length (in number of edges) of the longest path from $s$ to $t$.
The term {\em arithmetic} (or {\em algebraic}) {\em branching program} goes back at least to~\cite{Nisan91,BM99} but these objects were used implicitly much earlier, for instance in~\cite{Valiant79}.
Skew circuits, weakly skew circuits and arithmetic branching programs are
essentially equivalent models. Indeed, as shown in~\cite{KaKoi08} 
they simulate each other with only linear overhead (see~\cite{Jansen08} for
the multilinear case).

\section{From Circuits to Branching Programs} \label{tobp}

We first recall Lemma~4 from~\cite{MP08}.
\begin{lemma} \label{lemma4}
Let $C$ be a circuit of size $t$ and formal degree $d$, containing only binary unweighted arithmetic gates. There exists a weakly skew circuit 
$C'$ of formal degree $d$ and size at most $t^{\log 2d}$ which computes the same polynomial. 
\end{lemma}
 The fact that $C'$ has same formal degree as $C$ is not explicitly stated in~\cite{MP08}, but it can be checked that their construction does satisfy this additional property (more on this in the proof of Proposition~\ref{tows}).
We would like to apply this construction not to $C$ itself, but to a ``normal form'' of $C$ containing weighted addition gates.
We begin with an easy lemma.
\begin{lemma} \label{addcircuit}
Let $C$ be a circuit made only of input gates  and (ordinary) addition 
or subtraction gates.
 Each gate of $C$ 
is equivalent to a
weighted addition gate of total weight at most $2^s$, where $s$ is the number
of arithmetic gates in $C$.
\end{lemma}
\begin{proof}
By induction on $s$. The result is true for $s \leq 1$ since an input gate can be viewed as a unary weighted addition gate of weight~1, and an ordinary addition or subtraction gate as a binary weighted addition gate of weight~2. 
For $s>1$, consider an addition or subtraction gate which is an output of $C$.
By induction hypothesis each of the two inputs of the gate computes
a function of the form $\sum_{i=1}^n a_ix_i$ where $x_1,\ldots,x_n$ are the inputs of $C$ and $\sum_i |a_i| \leq 2^{s-1}$.
Therefore the output gate computes a function of the same form with total 
weight at most $2^s$.
\end{proof}
\begin{lemma} \label{binary2weighted}
Let $C$ be a circuit containing $s$ (weighted) addition gates and $m$ multiplication gates. There is an equivalent circuit $C_{+}$ such that:
\begin{itemize}
\item[(i)] $C_{+}$ contains at most $s$ addition gates and $m$ 
multiplication gates.
\item[(ii)] Any input to an addition gate is an input of $C_{+}$ or the output of a multiplication gate (in other words, the output of an addition gate can be fed only to 
multiplication gates).
\item[(iii)] If all the addition gates of $C$ are ordinary additions or subtractions,  the total weight of 
every addition gate of $C_+$  is at most $2^s$.
\item[(iv)] $C^+$ is of same formal degree as $C$.
\end{itemize}
\end{lemma}
In this lemma and elsewhere in the paper, ``equivalent'' means that $C_{+}$ computes
the same polynomial as $C$.
\begin{proof}[of Lemma~\ref{binary2weighted}]
We will keep the same multiplication gates in $C_{+}$ as in $C$.
Consider a multiplication gate in $C$ having 
at least one addition gate $\gamma$ as an input. We can view $\gamma$
as the output of a maximal subcircuit which does not contain any internal
multiplication gate (the inputs to the subcircuit are therefore inputs 
of $C$ or multiplication gates). 
The output of this subcircuit is a linear function of its inputs.
We can therefore replace the subcircuit by a single (weighted) addition gate 
$\gamma'$.
Moreover, in the case where all the addition gates of $C$ are ordinary additions or subtractions, $\gamma'$ 
can be taken of weight 
at most $2^s$ by Lemma~\ref{addcircuit}. 

 We perform this replacement simultaenously for all
addition gates of $C$ feeding into a multiplication gate.
If the output of $C$ is a multiplication gate, we are done.
If the output is an addition gate, we likewise replace its maximal subcircuit
by a weighted addition gate.
The resulting circuit $C^+$ satisfies properties (i), (ii) and (iii).

A straightforward induction shows that every multiplication gate $\alpha$ of $C$ has same formal degree as the corresponding gate in $C^+$; and that 
if $\alpha$ has an input $\gamma$ which is an addition gate, the formal degree
of the corresponding gate $\gamma'$ in $C^+$ will be equal to that of $\gamma$.
Hence property (iv) is satisfied as well.
\end{proof}
The same transformation as in Lemma~\ref{lemma4} can be applied to $C_{+}$
instead of $C$. The resulting weakly skew circuit contains weighted 
addition gates.
\begin{proposition} \label{tows}
Let $C$ be a circuit of size $t$ and formal degree $d$ where all multiplication gates are binary.
There exists a weakly skew circuit 
$C'$ of degree $d$ and size at most $t^{\log 2d}$ which computes the same polynomial. 
In $C'$, any input to an addition gate is an input of the circuit or the output of a multiplication gate. 
Moreover, if all the addition gates of $C$ are ordinary additions or subtractions,  the total weight of 
every addition gate of $C'$  is at most $2^t$.
\end{proposition}
\begin{proof}
We only give a sketch since this is really the same construction 
as in Lemma~4 of~\cite{MP08}. 
We briefly explain below why this construction preserves properties
(ii), (iii) and (iv) from Lemma~\ref{binary2weighted}, and refer 
to~\cite{MP08} for more details.
To achieve weak skewness $C'$ contains multiple copies of each gate
of $C^+$. Moreover, the connection pattern of $C^+$ is preserved in the
following sense. If $\alpha'$ is a copy of a multiplication gate
$\alpha$ then its two inputs $\beta'$ and $\gamma'$ are copies of the two
inputs $\beta$ and $\gamma$ of $\alpha$. Likewise, for any addition gate $\alpha$ of $C^+$ the inputs of a copy $\alpha'$  will be copies of its inputs,
 and moreover $\alpha'$ and $\alpha$ will have the same weights 
(\cite{MP08} considers
only unweighted binary addition gates, but the general case is identical).
In particular, $\alpha$ and $\alpha'$
 have same total weight and the inputs to $\alpha'$
are inputs of $C'$ or multiplication gates. A straightforward induction shows
that every gate of $C^+$ has same formal degree as its copies in $C'$.
\end{proof}
\begin{proposition} \label{skew2abp}
Let $C$ be a weakly skew circuit of size $m$ and formal degree $d$, 
with weights of addition gates coming from some set $W$. 
Assume moreover that any input to an addition gate is an input of $C$ or the output of a multiplication gate.
There exists an equivalent arithmetic branching program $G$ of size 
at most $m+1$ and depth at most $3d-1$. The edges of $G$ are labeled by 
inputs of $C$
or constants from $W$.
\end{proposition}
\begin{proof}
The construction is  similar to that of (\cite{MP08}, Lemma~5).
The main new  point is to check the depth bound. 
Recall from Section~\ref{arith} that for every multiplication gate $\alpha$
in $C$ we have an {\em independent subcircuit} which is connected to the
remainder of $C$ only by the arrow from the subcircuit to $\alpha$.
As in~\cite{MP08} we say that a gate is reusable if it does not belong to 
any independent subcircuit.
Also as in~\cite{MP08}, we will
prove a version of Proposition~\ref{skew2abp}
for circuits with multiple outputs. 

We will show by induction that for any reusable gate $\alpha$ of $C$
there is a vertex $t_{\alpha}$ in $G$ such that the weight of $(s,t_{\alpha})$ 
is the polynomial computed by $\alpha$.
As to the depth, we will show that if $\alpha$ is an addition gate computing 
a polynomial of formal degree $d_{\alpha}$, the
depth of $t_{\alpha}$ in $G$ (the length of the longest path from $s$ to $t_{\alpha}$) is at most $3d_{\alpha}-1$; if $\alpha$ is a multiplication
or input gate, its depth is at most $3d_{\alpha}-2$.

The beginning of the induction is clear: a weakly skew circuit $C$ 
of size $m=1$ 
is reduced to a single gate $\alpha$ labeled by some input $x$.
The corresponding graph $G$ has two nodes $s$ and $t$, with an edge
from $s$ to $t$ labeled by $x$. We take of course $t_{\alpha}=t$.
We have $d_{\alpha}=1$, and this gate is indeed at depth $3d_{\alpha}-2=1$.

Consider now a weakly skew circuit $C$ 
of size $m \geq 2$, and let $\alpha$ be one of 
its ouptut gates. Removing $\alpha$ from $C$, we obtain a circuit $C'$
of size $m-1$. By induction hypothesis, there is a corresponding graph
$G'$ of size at most $m$ with a distinguished vertex $s$.

If $\alpha$ is an input gate labeled by $x$, we obtain $G$ by adding 
a vertex $t_{\alpha}$ to $G'$, and an edge from $s$ to $t_{\alpha}$ labeled 
by $x$.

Assume now that $\alpha$ is a (weighted) addition gate, with $k$ 
(distinct) inputs $\alpha_1,\ldots,\alpha_k$. These $k$ gates must
be reusable, so by induction hypothesis we have vertices $t_{\alpha_i}$ 
in $C'$ so that the weight of $(s,t_{\alpha_i})$ is equal to the polynomial
computed by $\alpha_i$. Moreover, since $\alpha$ is an addition gate
the $\alpha_i$ are multiplication or input gates, and are therefore
at depth at most $3d_{\alpha_i}-2 \leq 3d_{\alpha}-2$.
We obtain $G$ by adding a new vertex $t_{\alpha}$ to $G'$, 
and $k$ new edges from the $t_{\alpha_i}$ to $t_{\alpha}$ (labeled by the
same weights as the incoming edges of the addition gate~$\alpha$).
The weight of $(s,t_{\alpha})$ in $G$ is clearly equal to the polynomial
computed by $\alpha$, and $t_{\alpha}$ is at depth at most 
$(3d_{\alpha}-2)+1=3d_{\alpha}-1$.

Assume finally that $\alpha$ is a multiplication gate with inputs 
$\beta$ and $\gamma$. Let $C_{\beta}$ and $C_{\gamma}$ be the corresponding 
subcircuits. Since $C$ is weakly skew, one of the two subcircuits
(say, $C_{\gamma}$) is independent from the rest of $C$. 
Hence $m=m_{\beta}+m_{\gamma}+1$ where $m_{\beta}$ and $m_{\gamma}$ are the sizes of $C_{\beta}$ and $C_{\gamma}$. We can apply separately the induction
hypothesis to $C_{\beta}$ and $C_{\gamma}$. This yields two graphs
$G_{\beta}$ and $G_{\gamma}$ of respective sizes at most $m_{\beta}+1$ and 
$m_{\gamma}+1$, with sources $s_{\beta}$
and $s_{\gamma}$. 
In these graphs there are vertices $t_{\beta}$ and $t_{\gamma}$ such 
that the weight of $(s_{\beta},t_{\beta})$ in $C_{\beta}$ is equal 
to the polynomial computed by gate $\beta$, 
and the weight of $(s_{\gamma},t_{\gamma})$ in $C_{\gamma}$ is equal 
to the polynomial computed by gate $\gamma$.
We construct $G$ from these two graphs by identifying $t_{\beta}$
and $s_{\gamma}$. The source of $G$ is $s=s_{\beta}$.
This graph is of size at most $(m_{\beta}+1)+(m_{\gamma}+1)-1=m \leq m+1$.
In $G$, the vertex associated to gate $\alpha$ will be $t_{\alpha}=t_{\gamma}$.
The weight of $(s,t_{\gamma})$ in $G$ is indeed 
equal to the polynomial computed
by gate $\alpha$. For vertices $v$ in $G_{\gamma}$ the weight of $(s,v)$
in $G$ is {\em not} equal to the weight of $(s_{\gamma},v)$ in $G_{\gamma}$,
but as pointed out in~\cite{MP08} this does not matter since these vertices
correspond to non-reusable gates of $C$.

Let $d$, $d_{\beta}$ and $d_{\gamma}$ be the formal degrees of the circuits
$C$, $C_{\beta}$ and $C_{\gamma}$. 
By induction hypothesis, $t_{\gamma}$ is at depth at most $3d_{\gamma}-1$ in
$G_{\gamma}$, and $t_{\beta}$ is at depth at most $3d_{\beta}-1$ in
$G_{\beta}$. In $G$, $t_{\gamma}$ is therefore at depth 
at most $(3d_{\beta}-1)+(3d_{\gamma}-1)=3d-2$.
\end{proof}
Combining Propositions~\ref{tows} and~\ref{skew2abp} yields the following result.
\begin{theorem} \label{circuit2abp}
Let $C$ be a circuit of size $t$ and formal degree $d$ where all multiplication gates are binary. There is an equivalent 
arithmetic branching program $G$ of size at most $t^{\log 2d}+1$ and depth 
at most $3d-1$.  The edges of $G$ are labeled by 
inputs of $C$ or by constants. 
Moreover, if all the addition gates of $C$ are ordinary additions or subtractions then these constants are integers of absolute value at most $2^t$.
\end{theorem}

\section{From Branching Programs to Depth-4 Circuits} \label{todepth4}

In this section we complete the reduction to circuits of depth 4.
\begin{lemma} \label{abp2power}
Let $G$ be an arithmetic branching program of size $m$ and depth $\delta$,
with 
edges labeled by elements from some set $S$. 
There is an $m \times m$ matrix $M$ such that the polynomial computed by $G$
is equal to the entry at row 1 and column $m$ of the matrix power $M^{p}$,
for any integer $p \geq \delta$.
Moreover, the entries of $M$ are in the set $S \cup \{0,1\}$.
\end{lemma}
\begin{proof}
Fix a topological ordering of the nodes of $G$, with the source $s$ labeled
1 and the target $t$ labeled $m$. We define $M$ as the adjacency matrix
of the graph $G'$ obtained from $G$ by adding a loop of weight 1 on vertex 
$t$. In other words, $M_{mm}=1$ and in all other cases $M_{ij}$ is the
(possibly null) weight from node $i$ to node $j$ of $G$. Note that $M$ 
is upper-diagonal, with all diagonal entries equal to 0 except $M_{mm}$.
It follows from the classical relation between matrix powering and paths
in graphs that $(M^p)_{1m}$ is equal to the sum of weights of all $st$-paths
of length exactly $p$ in $G'$. This is also the sum of weights of all $st$-paths of length at most  $p$ in $G$, and for $p \geq \delta$ this is the
output of the arithmetic branching program.
\end{proof}
Note that for $p \geq \delta$ all entries of $M^p$ except $(M^p)_{1m}$ 
are equal to zero.

In the last step in our series of reduction, we explain 
(following basically the same strategy as in~\cite{AgraVinay08})
how to perform the matrix powering operation in the above lemma with 
depth four formulas, and also depth four circuits.
\begin{proposition} \label{abp2depth4}
Let $G$ be an arithmetic branching program of size $m$ and depth $\delta$.
There is an equivalent depth four circuit $\Gamma$ with $m^2+1$ unweighted addition gates
and $m^{\lceil\sqrt{\delta}\rceil+1}+m^{\lceil\sqrt{\delta}\rceil-1}$ multiplication gates.
There is also an equivalent depth four formula $\Gamma_f$ 
with $m^{\lceil\sqrt{\delta}\rceil-1}+1$ unweighted addition gates 
and  $m^{\lceil\sqrt{\delta}\rceil-1}+m^{2\lceil\sqrt{\delta}\rceil-2}$ 
multiplication gates.

The inputs  of $\Gamma$ and $\Gamma_f$ are from the set as the edge labels of $G$, and their multiplication gates are of fan-in $\lceil \sqrt \delta \rceil$.
\end{proposition}
\begin{proof}
We need to compute $M^p$, where $p \geq \delta$ and $M$ is 
as in Lemma~\ref{abp2power}. Let $p$ be the smallest square integer bigger or equal to $\delta$. From $M$ we will compute $N=M^{\sqrt{p}}$ 
by a depth 2 circuit $\Gamma_2$,
and then from $N$ we will compute $M^p=N^{\sqrt{p}}$ using the same circuit.
With a depth 2 circuit one cannot play clever tricks: we can only 
expand a polynomial as a sum of monomials. 
In this case we express each entry of $N$ as a sum of $m^{\sqrt{p}-1}$
products of length $\sqrt{p}$, by brute-force expansion of the product 
$M^{\sqrt{p}}$ . This yields a circuit $\Gamma_2$
with $m^2$ addition gates (one for each entry of $N$) and $m^{\sqrt{p}+1}$
multiplication gates. We can double those estimates to upper bound
the size of $\Gamma$. To arrive at the slightly better estimate in 
the statement of Proposition~\ref{abp2depth4}, note that the second copy
of $\Gamma_2$ only needs to compute a single entry of $N^{\sqrt{p}}$.

In order to obtain an arithmetic formula, we recompute from scratch each
entry of $N$ whenever it is used by the second copy of $\Gamma_2$.
The arithmetic formula therefore computes a sum of $m^{\sqrt{p}-1}$ products,
where each product is a sum of $m^{\sqrt{p}-1}$ products of entries of $M$.
We therefore have one addition and $m^{\sqrt{p}-1}$ products gates in the
top two levels,  $m^{\sqrt{p}-1}$ addition and $m^{2(\sqrt{p}-1)}$ multiplication gates in the two bottom levels.
\end{proof}
Note the significant saving in the number of addition gates if we use
depth four circuits instead of depth four formulas.
We can now prove our main depth reduction result.
\begin{theorem} \label{circuit2depth4}
Let $C$ be an arithmetic circuit of size $t$ and formal degree $d$ where all multiplication gates are binary.
There is an equivalent depth four circuit $\Gamma$ with at most 
$(t^{\log 2d}+1)^2+1$ unweighted addition gates and at most 
$2(t^{\log 2d}+1)^{\sqrt{3d}+2}$ multiplication gates.

There is an equivalent arithmetic formula $\Gamma_f$ of depth four with at most
$(t^{\log 2d}+1)^{\sqrt{3d}}+1$ unweighted addition gates and at most
$2(t^{\log 2d}+1)^{2\sqrt{3d}}$ multiplication gates.
The inputs of $\Gamma$ and $\Gamma_f$ 
are  inputs of $C$ or constants;
their multiplication gates are of fan-in at most $\sqrt{3d}+1$.

If $C$ is constant-free, and if all the addition gates of $C$ are ordinary additions or subtractions, then these constants are integers of absolute value at most $2^t$.
\end{theorem}
\begin{proof}
Combine Theorem~\ref{circuit2abp} and Proposition~\ref{abp2depth4}.
\end{proof}
\begin{remark} \label{highdepth}
We can obtain smaller circuits for $C$ by going for a constant depth larger
than four. Let $M$ be the matrix in the proof of Proposition~\ref{abp2depth4}.
To compute a power $M^p$ we can start from $M$ and raise   
our matrix to the power $\sqrt[\Delta]{p}$ repeatedly ($\Delta$ times).
If we implement each of the $\Delta$ powerings by a depth 2 circuit, 
we obtain for the branching program $G$ a circuit of depth $2\Delta$ 
and size $m^{O(\sqrt[\Delta]{p})}$, for any constant $\Delta \geq 2$.
For $C$, this translates into a circuit of depth $2\Delta$ and size 
$t^{O(\sqrt[\Delta]{d} \log d)}$.
\end{remark}
If we start from arithmetic formulas (or more generally weakly skew circuits)
 instead of general arithmetic circuits,
we can obtain depth four formulas and circuits of smaller
size than in Theorem~\ref{circuit2depth4}. 
Indeed, in this case 
we do not need the transformation from arithmetic circuits to 
weakly skew circuits given by Proposition~\ref{tows}.
This saves a factor of roughly $\log 2d$ in the exponents of our complexity bounds.
\begin{theorem}
Let $C$ be a weakly skew circuit of size $t$ and formal degree $d$.
There is an equivalent depth four circuit $\Gamma$ with at most 
$(t+1)^2+1$ unweighted addition gates and at most 
$2(t+1)^{\sqrt{3d}+2}$ multiplication gates.

There is an equivalent arithmetic formula $\Gamma_f$ of depth four with at most
$(t+1)^{\sqrt{3d}}+1$ unweighted addition gates and at most
$2(t+1)^{2\sqrt{3d}}$ multiplication gates.
The inputs of $\Gamma$ and $\Gamma_f$ 
are  inputs of $C$ or constants;
their multiplication gates are of fan-in at most $\sqrt{3d}+1$.

If $C$ is constant-free, and if all the addition gates of $C$ are ordinary additions or subtractions, then these constants are integers of absolute value at most $2^t$.
\end{theorem}
\begin{proof}
Before applying Proposition~\ref{skew2abp} we make sure that any input to
an addition of gate of $C$ is an input of the circuit or a multiplication 
gate. By Lemma~\ref{binary2weighted} this condition can be ensured
without increasing the size of $C$ (and this transformation preserves
weak skewness). Hence there is an equivalent arithmetic branching program of 
size at most $t+1$ and depth at most $3d-1$.
Then we convert this branching program 
into a depth 4 circuit or a depth 4 formula using Proposition~\ref{abp2depth4}.

When $C$ is constant free, the bound on the absolute value of the constants
of $\Gamma$ and $\Gamma_f$ comes (as in Theorem~\ref{circuit2depth4})
from property~(iii) in Lemma~\ref{binary2weighted}.
\end{proof}
The savings in the number of addition gates 
in depth four circuits compared to depth four formulas 
are especially significant 
in the above theorem:  our circuits contain only quadratically many
addition gates. This is a relevant parameter since the 
number of addition gates (minus 1) is equal to the number of {\em distinct}
sparse polynomials in a sum of products of sparse polynomials~\cite{Koi10a}.

\section{Depth Reduction for $\vp$} \label{vp}

In accordance with Definition~\ref{weighted},
a unary weighted addition gate outputs~$\alpha \cdot  x$, 
where $\alpha$ is the weight of the gate and $x$ its input. 
Recall also from the definition of formal degree in Section~\ref{arith}
that the formal degree of such a gate is equal to that of its 
input.

The following result is essentially Lemma 2 from~\cite{Malod03},
 written in a different language. We give the proof because 
we will build on it in the next section.
\begin{proposition} \label{vpdegred}
Any $\vp$ family $(f_n)$ can be computed by a polynomial-size family $(C_n)$ 
of circuits of formal degree $\deg(f_n)$. 
The addition gates of $C_n$ are unary weighted
or binary unweighted (i.e., ``ordinary'').
\end{proposition}
\begin{proof}
Since $(f_n)$ is in $\vp$, this family can be computed by a family $(C'_n)$
of arithmetic circuits of polynomial size where all the arithmetic gates
are binary unweighted.
To construct $C_n$ from $C'_n$ we use a small variation on the standard
homogenization trick. In order to homogenize $C'_n$ one would normally 
represent each gate $\gamma$ 
computing a polynomial $f_{\gamma}$
by a sequence $\gamma_i$ of $d_n+1$ gates, where 
$i$ ranges from 0 to $d_n$ and $\gamma_i$ computes 
the homogenous component of $f_{\gamma}$ of degree $i$.
The homogenous components of degree higher than $d_n$ can be discarded
since they cannot contribute to the final output. This construction 
preserves polynomial circuit size, and each gate now computes
a polynomial of degree at most $d_n$. But formal degree can be higher
due to multiplication by constants 
(i.e., homogenous components of degree 0).

To circumvent this difficulty, we get rid of the gates $\gamma_0$ 
representing homogenous components of degree 0. We will therefore
construct a circuit $C''_n$ which computes the sum of all homogenous
components of $f_n$ of degree at least 1. Our final circuit $C_n$
will then add the output of $C''_n$ to the constant term of $f_n$,
at the cost of one additional arithmetic operation.

We will use unweighted addition gates inside $C''_n$.
Indeed, let $\gamma$
be a multiplication gate of $C_n$ with inputs $\alpha$ and $\beta$. 
To obtain $f_{\gamma,i}$, the homogenous component of degree $i$, 
one normally writes $f_{\gamma,i}=\sum_{j=0}^{i} f_{\alpha,j}f_{\beta,i-j}$.
This expression involves $ f_{\alpha,0}$ and $f_{\beta,0}$, which as 
we have said are not represented by any gate of $C''_n$.
Therefore, to compute e.g. $f_{\alpha,0} f_{\beta,i}$,
instead of a multiplication gate we use a unary addition gate with
input $f_{\beta,i}$ and weight  $f_{\alpha,0}$.
A straightforward induction shows that a gate $\gamma_i$ in $C''_n$ will
have formal degree~$i$. As a result, $C''_n$ and $C_n$ will be
of formal degree $d_n$.
\end{proof}

\begin{theorem} \label{vp2depth4}
Let $(f_n)$ be a $\vp$ family of polynomials of degree $d_n=\deg(f_n)$.
This family can be computed by a family $(\Gamma_n)$ of depth four circuits
with $n^{O(\log d_n)}$ addition gates and $n^{O(\sqrt{d_n}\log d_n)}$
multiplication gates.
The family $(f_n)$ can also be computed by a family $(F_n)$ of depth four
arithmetic formulas of size $n^{O(\sqrt{d_n}\log d_n)}$.
The inputs to $\Gamma_n$ and $F_n$ are variables of $f_n$ 
or constants; their multiplication gates are of fan-in at most
$\sqrt{3d_n}+1$.
\end{theorem}
\begin{proof}
This is an application of Theorem~\ref{circuit2depth4}: $t$ is polynomial 
in $n$, and by Proposition~\ref{vpdegred} we can take $d=d_n$.
\end{proof}

\section{Depth Reduction for $\vpzero$} \label{vp0}

We first show that a circuit of small size and degree where all inputs are in $\{-1,0,1\}$ cannot compute a large integer.
\begin{lemma} \label{constantsize}
Let $C$ be a constant-free and variable-free circuit of size 
$t$ and formal degree $d$ where all arithmetic gates are binary unweighted. The output of $C$ is an integer of absolute value at 
most $2^{td}$.
\end{lemma}
\begin{proof}
By induction on $t$. For $t=1$ the circuit contains a single input gate, 
which must carry an integer in $\{-1,0,1\}$. The result is therefore 
true for $t=1$. Consider now a circuit $C$ of size $t \geq 2$,
and let $d_1$ and $d_2$ be the formal degrees of the two inputs
to the output gate.
By induction hypothesis these two gates carry integers of absolute
value at most $2^{(t-1)d_1}$ and $2^{(t-1)d_2}$.
If the output gate is an addition we have $d_1, d_2 \leq d$ and 
$C$ therefore computes an integer of absolute value at most 
$2^{(t-1)d}+2^{(t-1)d} \leq 2^{td}$. If the output gate is a multiplication,
we have $d=d_1+d_2$ and $C$ computes  an integer of absolute value at most 
$2^{(t-1)d_1} \times 2^{(t-1)d_2} \leq 2^{td}$.
\end{proof}

\begin{proposition} \label{degred}
Any $\vp^0$ family $(f_n)$ can be computed by a family $(C_n)$ 
of constant-free circuits of polynomial size and formal 
degree $\deg(f_n)$. The arithmetic gates of $C_n$ are binary multiplication,
ordinary addition or subtraction gates.
\end{proposition}
\begin{proof} 
Since $(f_n)$ is in $\vp^0$, this family can be computed by a family $(C'_n)$
of constant-free circuits of polynomial size and polynomial formal degree.
All the arithmetic gates of $C'_n$ can be assumed to be binary unweighted.
To construct $C_n$ from $C'_n$ we proceed along the same lines as in Proposition~\ref{vpdegred}.
In particular, we will again construct a circuit $C''_n$ 
which computes the sum of all homogenous
components of $f_n$ of degree at least 1. Our final circuit $C_n$
then adds the output of $C''_n$ to the constant term of $f_n$ (call it
$c_n$).
By Lemma~\ref{constantsize}, $c_n$ has polynomial bit size
(it is equal to the output of $C'_n$ when all variables are set to 0).
We can therefore compute $|c_n|$ 
from scratch using a sequence of multiplications
by 2 and additions of bits.
We use an addition to perform a multiplication by 2, so this construction
does not require any multiplication gate. 
Finally, depending on the sign of $c_n$ we add or subtract $|c_n|$ to the output of $C''_n$.
The resulting circuit $C_n$ will have same formal 
degree as $C''_n$.

We also need to use a similar trick inside $C''_n$. Indeed, let $\gamma$
be a multiplication gate of $C_n$ with inputs $\alpha$ and $\beta$. 
To obtain $f_{\gamma,i}$, the homogenous component of degree $i$, 
one normally writes $f_{\gamma,i}=\sum_{j=0}^{i} f_{\alpha,j}f_{\beta,i-j}$.
This expression involves $ f_{\alpha,0}$ and $f_{\beta,0}$, which as 
explained in the proof of Proposition~\ref{vpdegred} are not represented by any gate of $C''_n$.
Therefore, to compute e.g. $f_{\alpha,0} f_{\beta,i}$ we start from 
$f_{\beta,i}$ and compute the product using a sequence of multiplications
by 2 and additions of $f_{\beta,i}$. As explained above, 
thanks to Lemma~\ref{constantsize} this can be done with a polynomial
number of addition gates, at most one subtraction and no multiplication gate. 
A straightforward induction shows that a gate $\gamma_i$ in $C''_n$ will
have formal degree $i$. As a result, $C''_n$ and $C_n$ will be 
of formal degree $d_n$.
\end{proof}
By Proposition~\ref{nosub}, one can get rid of the subtraction gates in Proposition~\ref{degred} at the cost of a linear increase in circuit size
and an increase in the formal degree by just 1
(using Lemma 3 from~\cite{Malod03} instead of Proposition~\ref{nosub}
would give a worse degree bound).


\begin{theorem} \label{vpzero2depth4}
Let $(f_n)$ be a $\vpzero$ family of polynomials of degree $d_n=\deg(f_n)$.
This family can be computed by a family $(\Gamma_n)$ of depth four circuits
with $n^{O(\log d_n)}$ addition gates and $n^{O(\sqrt{d_n}\log d_n)}$
multiplication gates.
The family $(f_n)$ can also be computed by a family $(F_n)$ of depth four
arithmetic formulas of size $n^{O(\sqrt{d_n}\log d_n)}$.
The inputs to $\Gamma_n$ and $F_n$ are variables of $f_n$ 
or relative integers 
of polynomial bit size; their multiplication gates are of fan-in at most
$\sqrt{3d_n}+1$.
\end{theorem}
\begin{proof}
This is an application of Theorem~\ref{circuit2depth4}: $t$ is polynomial 
in $n$, and by Proposition~\ref{degred} we can take $d=d_n$.
\end{proof}

\section{Application to Boolean Circuits}

In this section we give an application of our results to boolean circuit
complexity. A discussion of depth reduction in the boolean versus arithmetic 
setting can already be found in ~\cite{AgraVinay08}, but that paper
did not actually provide any result of this type.
Here we use arithmetic techiques to reprove a known result~:
 languages in  $\logcfl$ have nontrivial constant-depth circuits.
\begin{proposition} \label{logcfl}
Let $L$ be a languange in $\logcfl$. For every $\epsilon>0$, $L$ can be decided
by a family of constant-depth circuits $\Gamma_n$ of size $2^{n^{\epsilon}}$.
The gates of $\Gamma_n$ are OR or AND gates, both of unbounded fan-in, 
and NOT gates.
\end{proposition}
\begin{proof}[Sketch]
It is known that languages in $\logcfl$ can be recognized by families
$(C_n)$ of semi-unbounded circuits of logarithmic depth and polynomial
size~\cite{Ven91}. Each circuit $C_n$ has $2n$ inputs; the remaining gates are
AND gates of fan-in 2 or OR gates of unbounded fan-in.
A language $L$ in $\logcfl$ is recognized by the corresponding circuit 
family in the following sense: a word $x \in \{0,1\}^n$ belongs to $L$
iff the input $x_1\ldots x_n \overline{x_1} \ldots \overline{x_n}$
is accepted by $C_n$.

We view $C_n$ as an arithmetic circuit over the boolean semiring
 ${\cal R}=(\{0,1\},\vee,\wedge)$: the boolean OR is the addition of $\cal R$,
and the boolean AND is its multiplication. 
The semi-unboundedness property together with the $O(\log n)$ depth bound 
imply that $C_n$ is of polynomially bounded formal degree.
It follows that we can apply the results of Section~\ref{todepth4} 
(up to now we have considered only arithmetic circuits over fields,
but the main results and their proofs apply to semirings).
The existence of a suitable constant-depth circuit family $(\Gamma_n)$ 
therefore follows from Remark~\ref{highdepth}. 
Note that the depth of $\Gamma_n$ depends on the exponent in the polynomial
bound for the formal degree of $C_n$.
\end{proof}
\begin{remark}
Instead of working over the semiring $\cal R$ in the above proof, 
one could also work over
$(\nn,+,\times)$. To do this replace each OR gate of $C_n$ by a $+$ gate and each AND  gate by a $\times$ gate; apply Remark~\ref{highdepth} to 
the resulting circuit; and finally convert back addition gates into 
OR gates and multiplication gates into AND gates.
\end{remark}
One can find in Lemma~8.1 of~\cite{AHMPS} 
a proof of Proposition~\ref{logcfl} for languages in $\nl$
(a subset of $\logcfl$), and the authors observe that the proof also
applies to $\logcfl$. 
 According to~\cite{Viola09}, the
result for $\nl$ is usually credited to~Nepomnjascii~\cite{Nep70}.
Nepomnjascii proved a uniform version of this result which
in recent years has been used in time-space lower bounds (see~\cite{Melk07}
for a survey on this topic).
The result for languages in $\mathsf L$ was used in~\cite{GV04} to construct
certain uniform families of expanders.

Another depth reduction result due to Valiant shows that 
boolean circuits of linear size and depth $O(\log n)$ have depth-3 
circuits of size $2^{O(n/ \log \log n)}$.
This result is stated in~\cite{Valiant83} for monotone circuits. 
The statement for non-monotone circuits 
(and a proof based on~\cite{Valiant77,Valiant83}) 
can be found in~\cite{Viola09}.
All these results suggest that lower bounds on the size of
circuits of logarithmic depth might be obtained by proving strong enough lower 
bounds for constant-depth circuits (and quite possibly explain why
it is difficult to obtain very strong lower bounds for constant-depth circuits).

\section{Reduction to Polylogarithmic Depth} \label{polylog}

It was shown by Valiant, Skyum, Berkowitz and Rackoff~\cite{VSBR83} 
that arithmetic circuits of polynomially bounded  size and degree
can be transformed into circuits of polylogarithmic depth and polynomial size
(the depth can even be made logarithmic with addition gates of unbounded 
fan-in).
Since then several refinements of this fundamental result have been published, adressing in particular the issues of uniformity~\cite{MRK86,AJMV98} 
or multilinearity~\cite{RY08}.
In this section we give another proof of reduction to polylogarithmic depth.
The depth bound that we obtain is worse than~\cite{VSBR83} 
by a logarithmic factor. This result is therefore not new neither optimal, 
but nonetheless we feel that it is worth presenting here because its proof
is quite simple and based on the same tools as the remainder of the paper:
(weakly) skew circuits and arithmetic branching programs.

Before turning to general arithmetic circuits, we first parallelize arithmetic 
branching programs.
\begin{proposition} \label{abp2logdepth}
Let $G$ be a (multi-output) arithmetic branching program of size $m$ and depth $\delta$.
There is a multi-output arithmetic circuit $C$ of depth 
$2\lceil \log \delta \rceil$
which computes the $m$ polynomials represented by the $m$ nodes of~$G$.
The circuit contains $m^3  \lceil \log \delta \rceil$ 
binary multiplication gates and $m^2 \lceil \log \delta \rceil$ 
addition gates of unbounded fan-in.
\end{proposition}
\begin{proof}
It is again based on matrix powering. We start from the adjacency matrix 
of $G$, and add the identity matrix (instead of a single 1 on the diagonal
as in the proof of Lemma~\ref{abp2power}). Let $M$ be the resulting matrix.
Assuming again that the source node of $G$ is labeled 1, the polynomial
represented by node~$j$ of $G$ is equal to $(M^p)_{1j}$ for any power 
$p \geq \delta$. We will compute $M^p$ by repeated squaring. 
From $M$ we can compute $M^2$ by a depth 2 circuit with $m^3$ multiplication
gates and $m^2$ unbounded additions. We repeat this process 
$\lceil \log \delta \rceil$ times to obtain $M^p$.
\end{proof}

\begin{theorem}
Let $C$ be a circuit of size $t$ and formal degree $d$ where all multiplication gates are binary. 
There is an equivalent circuit $C'$ (with binary multiplication gates as well) 
of depth   $O(\log t \cdot \log d )$ and size $O(t^3 \log t \cdot \log d)$
\end{theorem}
\begin{proof}
We decompose $C$ in ``layers'' $C_i$: $C_i$ is made of all gates of $C$ 
of formal degree in the interval $[2^i,2^{i+1}[$. Here $i$ ranges from 0 to 
$\lfloor \log d \rfloor$. Each layer forms a (multi-output) arithmetic circuit;
for $i \geq 1$, the input gates of $C_i$ actually belong to previous 
$C_j$'s for various $j <i$. The crucial observation is that these 
arithemetic circuits are all skew, i.e., for each mutiplication gate
at least one of the two arguments is an input gate of $C_i$.
Indeed, the product of two  gates of formal degree at least $2^i$ is of
formal degree at least $2^{i+1}$ and therefore cannot belong to $C_i$.
But (as pointed out at the end of Section~\ref{arith}) skew circuits and arithmetic branching programs are essentially
equivalent objects. In particular, by Lemma~5 of~\cite{MP08} 
a skew circuit (or even a weakly skew circuit) 
of size $s$ can be simulated by an arithmetic branching program of size 
$s+1$ (this result of~\cite{MP08} is stated only for circuits with binary addition gates, but the proof clearly applies to unbounded fan-in 
as well\footnote{We gave in Proposition~\ref{skew2abp} a variation on this result.}). By Proposition~\ref{abp2logdepth} each $C_i$ is therefore 
equivalent to a circuit of depth $O(\log t)$ and size $O(t^3 \log t)$.
We multiply these estimates by $1+\lfloor \log d \rfloor$ to obtain the final result.
\end{proof}

\section*{Acknowlegments} I thank Eric Allender, Bruno Grenet, Natacha Portier and Amir Yehudayoff 
 for useful discussions on
this work. 


\end{document}